\documentclass[10pt]{article}

\usepackage[letterpaper, margin=0.85in]{geometry}
\usepackage{amsmath, amssymb, amsthm, mathtools}
\usepackage{bm}
\usepackage{enumitem}
\usepackage{booktabs}
\usepackage{array}
\usepackage{multirow}
\usepackage{graphicx}
\usepackage{caption}
\usepackage[round]{natbib}
\usepackage{hyperref}
\usepackage{xcolor}
\usepackage{titlesec}

\hypersetup{colorlinks=true, linkcolor=black, citecolor=black!70!blue, urlcolor=blue!60!black}
\titleformat{\section}{\large\bfseries}{\thesection}{0.8em}{}
\titleformat{\subsection}{\normalsize\bfseries}{\thesubsection}{0.7em}{}

\newtheorem{theorem}{Theorem}

\newtheorem{lemma}[theorem]{Lemma}
\newtheorem{corollary}[theorem]{Corollary}
\theoremstyle{definition}
\newtheorem{assumption}{Assumption}
\newtheorem{hypothesis}{Hypothesis}
\theoremstyle{remark}
\newtheorem{remark}{Remark}

\DeclareMathOperator{\op}{op}
\newcommand{\R}{\mathbb{R}}
\newcommand{\sisdr}{\mathrm{SI\text{-}SDR}}
\newcommand{\sig}{\widehat{s}}

\title{\vspace{-2.5em}
A Variational-Flow Analysis of Diffusion-Based Speech Enhancement under Noise-Power Mismatch}
\author{Shuubham Ojha\\\textit{University of Maryland, College Park}}
\date{}

\begin{document}
\maketitle
\vspace{-1.5em}

\begin{abstract}
Diffusion-based speech enhancement architectures pairing a deterministic predictor with a learned score network, such as StoRM~\citep{lemercier2023storm}, degrade with a visible corner in the SI-SDR curve when the noise amplitude is scaled away from training. We give a pathwise variational analysis of this phenomenon. Differentiating the reverse SDE with respect to the noise-scaling parameter $M$ yields an exact factorization of the output sensitivity, $\partial \sig^{(M)}/\partial M = K(M)\cdot \partial C_M/\partial M$, where $K(M)$ is determined by the score Jacobian along the reverse trajectory and $C_M = \Pi(y^{(M)})$ is the predictor output. Under three hypotheses on the reverse-process flow: continuity of the score Jacobian in $M$, continuity of the conditioning Jacobian in $M$, and non-degeneracy of $K(M^\ast)$, this factorization gives an iff statement: $M\mapsto \sig^{(M)}$ fails to be $C^1$ at $M^\ast$ if and only if $M\mapsto \Pi(y^{(M)})$ does. We prove the theorem and its discrete-sampler counterpart, and empirically substantiate each ingredient on a StoRM model trained on VoiceBank--DEMAND: the kink is present per-utterance and in the predictor-only output (piecewise log-linear fit $R^2 = 0.9994$ vs.~$0.9935$ for the best smooth fit); the score-only ablation is smooth; the score-network Jacobian is continuous in $M$ along the trajectory; and $K(M)$ is well-conditioned across the operating range, with $\sigma_{\min}(K(M))$ bounded well away from zero. The predictor stage is thus identified as the structural source of noise-power non-smoothness in this class of architectures.
\end{abstract}

\section{Introduction}

The pairing of a discriminative predictor with a score-based diffusion refinement stage has become a standard recipe for high-quality speech enhancement. In StoRM~\citep{lemercier2023storm}, a predictor $\Pi$ produces $C = \Pi(y)$ from the noisy observation $y$, and a reverse SDE is run on a learned score $s_\theta(x,t,C)$ with $C$ as conditioning and $X_T = C + \sigma_T \varepsilon$ as initialization. The combined system outperforms both stages in isolation on standard benchmarks~\citep{lemercier2023storm,richter2023sgmse}, and variants of this architecture form the state of the art in the field.

How such systems behave when the noise amplitude departs from the training condition is not well understood theoretically. Empirically the picture is striking: at the per-utterance level, the SI-SDR degradation curve as a function of noise scaling has a visible corner at the training amplitude: the slope changes discontinuously across it. The corner is a per-utterance phenomenon; large-population averages smooth it out because per-utterance corner locations drift slightly around the training amplitude.

We ask which component of the architecture produces this corner. The answer we give is theoretical: differentiating the learned reverse process with respect to the noise-scaling parameter $M$ leaves the reverse dynamics dependent on $M$ only through the predictor output $C_M$, and parameter-dependence of ODE flows yields an exact multiplicative decomposition of the output sensitivity. The multiplicative structure sharpens the localization to an iff: any non-smoothness in $M \mapsto \sig^{(M)}$ can only originate in $M \mapsto \Pi(y^{(M)})$, provided the score-network Jacobians along the trajectory are continuous in $M$ and $K(M^\ast)$ does not annihilate the predictor's jump direction. Both conditions are directly measurable and hold for the model we study.

The analysis is distinct from the path-space KL machinery used for oracle-comparison bounds in diffusion theory~\citep{chen2022sampling,benton2024nearly}: those methods bound one path measure against another via Girsanov, producing additive predictor--score error decompositions at the population level. Our object is parametric, a single learned process indexed by $M$. Thus, the natural machinery is ODE parameter-dependence, which gives a pathwise, multiplicative decomposition of a specific functional of $X_0$. The two frameworks answer different questions.

\paragraph{Contributions.}
\begin{enumerate}[label=\arabic*.,leftmargin=*,itemsep=2pt,topsep=2pt]
\item An exact pathwise factorization $V_0 = K(M) \cdot \mathrm{d}C_M/\mathrm{d}M$ of the output sensitivity to $M$ (Lemma~\ref{lem:variational}).
\item An iff localization theorem attributing non-smoothness of the enhancement output to non-smoothness of the predictor (Theorem~\ref{thm:kink}), and its extension to the discrete Euler--Maruyama sampler (Corollary~\ref{cor:discrete}).
\item Empirical validation on a 32-channel StoRM model: direct observation of the per-utterance kink and the predictor-only kink; verification that the score-Jacobian is continuous in $M$; verification that $K(M)$ is well-conditioned with $\sigma_{\min}(K(M))$ bounded away from zero across the operating range.
\end{enumerate}


\section{Setup}
\label{sec:setup}

\subsection{Notation and the three-process formulation}
Let $s \in \R^d$ be a clean speech signal and $n \in \R^d$ additive noise. For $M > 0$ set
\begin{equation}
y^{(M)} := s + \sqrt{M}\, n,
\end{equation}
so $M = 1$ recovers the training-time noise amplitude. Both $s$ and $n$ are held fixed; all $M$-dependence is parametric. The predictor produces $C_M := \Pi(y^{(M)})$, and the reverse SDE is
\begin{equation}
\mathrm{d} X_t^{(M)} = b_\theta(X_t^{(M)}, t, C_M)\,\mathrm{d}t + g(t)\,\mathrm{d}W_t,\qquad t \in [T, 0],
\end{equation}
with drift $b_\theta(x,t,C) = f(x,t) - g(t)^2 s_\theta(x,t,C)$, initial condition $X_T^{(M)} = C_M + \sigma_T \varepsilon$, and $W_t$, $\varepsilon$ held fixed across $M$. The enhanced output is $\sig^{(M)} := X_0^{(M)}$.

We work with the SI-SDR curve $\Psi(M) := \sisdr(\sig^{(M)}, s)$ but state everything for $\sig^{(M)}$ itself; SI-SDR is smooth away from the origin and inherits kinks from $\sig^{(M)}$ by the chain rule.

\subsection{Assumptions}

\begin{assumption}[Single-channel $M$-dependence]
\label{ass:single-channel}
The drift $b_\theta$, diffusion coefficient $g(\cdot)$, schedule $\sigma_\cdot$, and Brownian motion $W_\cdot$ are $M$-independent as functions of their explicit arguments. The map $M \mapsto C_M$ is the sole channel through which $M$ enters the reverse dynamics.
\end{assumption}

For StoRM with \texttt{condition="both"} the score network receives $C = (y, \Pi(y))$ as a joint conditioning input, so both channels of $C$ depend on $M$: $y^{(M)}$ smoothly (as $\sqrt{M}\, n$), and $\Pi(y^{(M)})$ potentially not. The variational identity below still factors cleanly because the smooth channel cannot contribute non-smoothness; we return to this after Lemma~\ref{lem:variational}.

\begin{assumption}[Regularity]
\label{ass:regularity}
The drift $b_\theta$ is jointly $C^2$ in $(x, C)$, with $\nabla_x b_\theta$, $\nabla_C b_\theta$ and their derivatives bounded on the compact set traced out by $\{(X_t^{(M)}, t, C_M): t \in [0,T],\, M \in \mathcal{M}\}$ for any compact $\mathcal{M} \subset (0,\infty)$.
\end{assumption}

For NCSN++ with smooth activations, Assumption~\ref{ass:regularity} is a standard consequence of finite weights and smooth composition. Pillar~P1 (\S\ref{sec:p1}) directly confirms Jacobian boundedness along the trajectory in our experiments.

\section{Variational-flow analysis}
\label{sec:varflow}

Under Assumption~\ref{ass:regularity}, $M \mapsto X_t^{(M)}$ is $C^1$ wherever $M \mapsto C_M$ is $C^1$~\citep[Ch.~V]{hartman2002ode}. Define the sensitivity $V_t := \partial X_t^{(M)}/\partial M$.

\begin{lemma}[Variational identity]
\label{lem:variational}
Under Assumptions~\ref{ass:single-channel} and \ref{ass:regularity}, $V_t$ satisfies
\begin{equation}
\label{eq:var-ode}
\frac{\mathrm{d}V_t}{\mathrm{d}t} = \nabla_x b_\theta(X_t^{(M)},t,C_M)\, V_t + \nabla_C b_\theta(X_t^{(M)},t,C_M)\, \frac{\mathrm{d}C_M}{\mathrm{d}M},
\end{equation}
with terminal condition $V_T = \mathrm{d}C_M/\mathrm{d}M$. Let $\Phi(t,\tau)$ denote the state-transition matrix of the homogeneous part along the trajectory, with $\Phi(\tau,\tau)=I$. Then
\begin{equation}
\label{eq:duhamel}
V_0 = \underbrace{\left[\Phi(0,T) + \int_T^0 \Phi(0,\tau)\, \nabla_C b_\theta(X_\tau^{(M)},\tau,C_M)\,\mathrm{d}\tau \right]}_{K(M)} \cdot \frac{\mathrm{d}C_M}{\mathrm{d}M}.
\end{equation}
\end{lemma}

\begin{proof}
Differentiating the reverse SDE with respect to $M$ and applying Assumption~\ref{ass:single-channel}: the Brownian increments are $M$-independent, $b_\theta$ has no explicit $M$-argument, and the chain rule gives \eqref{eq:var-ode}. The terminal condition follows from differentiating $X_T^{(M)} = C_M + \sigma_T \varepsilon$. Applying variation of constants to \eqref{eq:var-ode} with this terminal condition gives \eqref{eq:duhamel}; both the homogeneous and particular parts factor $\mathrm{d}C_M/\mathrm{d}M$ on the right.
\end{proof}

The factorization $\partial \sig^{(M)}/\partial M = K(M) \cdot \partial C_M/\partial M$ is exact and pathwise. $K(M)$ depends only on the score Jacobian $\nabla_x s_\theta$ along the trajectory (through $\Phi$) and the conditioning Jacobian $\nabla_C b_\theta$ along the trajectory. The predictor enters through the boundary term and through $\mathrm{d}C_M/\mathrm{d}M$; it does not appear inside $K(M)$ at all.

\begin{remark}[Multi-channel conditioning]
When $C = (y, \Pi(y))$, the map $M \mapsto y^{(M)} = s + \sqrt{M}\, n$ is $C^\infty$ everywhere on $M > 0$, so $\mathrm{d}C_M/\mathrm{d}M$ has one smooth component (the $y$-channel) and one potentially non-smooth component (the $\Pi(y)$-channel). $K(M)$ acts on the joint vector; the theorem below then localizes any $C^1$ failure to the predictor channel because the $y$-channel cannot contribute one.
\end{remark}

\begin{remark}[Multiplicative vs.\ additive]
The factorization in \eqref{eq:duhamel} differs structurally from oracle-comparison bounds for diffusion models~\citep{chen2022sampling}. Those compare two path measures (oracle $X^\ast$, learned $X_\theta$) via a Girsanov identity, producing additive predictor+score error decompositions of a KL divergence between path measures. Here there is no oracle, and there is a single learned process parametrized by $M$; the decomposition is by role. The score network appears as a coefficient inside the flow through $\Phi$, and the predictor appears as boundary data and conditioning. It is this multiplicative separation that gives an iff localization; an additive bound would give only one-sided control and could not distinguish the origin of a kink from its propagation.
\end{remark}

\section{Kink localization}
\label{sec:thm}

A kink in $\sig^{(M)}$ at $M^\ast$ is a jump in $V_0$ across $M^\ast$. Three hypotheses about the reverse-process flow determine what $K(M)$ does to such a jump.

\begin{hypothesis}[Score-Jacobian continuity]
\label{H:1}
$M \mapsto \nabla_x s_\theta(X_t^{(M)}, t, C_M)$ is continuous along the reverse trajectory for every $t \in [0,T]$.
\end{hypothesis}

\begin{hypothesis}[Conditioning-Jacobian continuity]
\label{H:2}
$M \mapsto \nabla_C b_\theta(X_t^{(M)}, t, C_M)$ is continuous along the reverse trajectory for every $t \in [0,T]$.
\end{hypothesis}

\begin{hypothesis}[Non-degeneracy of $K(M^\ast)$]
\label{H:3}
$K(M^\ast)$ is non-singular along the direction of any one-sided limit of $\mathrm{d}C_M/\mathrm{d}M$ at $M^\ast$.
\end{hypothesis}

\begin{theorem}[Kink localization]
\label{thm:kink}
Under Assumptions~\ref{ass:single-channel}--\ref{ass:regularity} and Hypotheses~\ref{H:1}--\ref{H:3},
\[
M\mapsto \sig^{(M)} \text{ fails to be } C^1 \text{ at } M^\ast \iff M \mapsto \Pi(y^{(M)}) \text{ fails to be } C^1 \text{ at } M^\ast.
\]
\end{theorem}

\begin{proof}
The identity $V_0 = K(M) \cdot \mathrm{d}C_M/\mathrm{d}M$ from Lemma~\ref{lem:variational} is the starting point of both directions.

\emph{Continuity of $K$.} Under H\ref{H:1}, $\nabla_x s_\theta$ is continuous in $M$ along the trajectory, so $\nabla_x b_\theta = \nabla_x f - g^2 \nabla_x s_\theta$ is too. The state-transition matrix $\Phi(0,\tau)$ solves a linear ODE whose coefficient depends continuously on $M$; standard parameter-dependence of ODEs~\citep[Ch.~V]{hartman2002ode} gives continuity of $\Phi(0,\tau)$ in $M$ for each $\tau$. Under H\ref{H:2}, $\nabla_C b_\theta$ is continuous in $M$. The integral $\int_T^0 \Phi(0,\tau)\, \nabla_C b_\theta\,\mathrm{d}\tau$ is then continuous in $M$ as an integral of a uniformly bounded continuous integrand (boundedness by Assumption~\ref{ass:regularity} on compact $M$-range). Hence $K(M)$ is continuous.

\emph{$(\Rightarrow)$ Kink in output $\Rightarrow$ kink in predictor.} Suppose $M \mapsto \sig^{(M)}$ fails to be $C^1$ at $M^\ast$: $V_0$ is discontinuous at $M^\ast$. Matrix--vector multiplication is jointly continuous (bilinear). So if both $K(M)$ and $\mathrm{d}C_M/\mathrm{d}M$ were continuous at $M^\ast$, their product $V_0$ would be too. Since $K$ is continuous at $M^\ast$ by the argument above but $V_0$ is not, $\mathrm{d}C_M/\mathrm{d}M$ must be discontinuous at $M^\ast$. Equivalently, $\Pi(y^{(M)})$ fails to be $C^1$ at $M^\ast$. This direction uses only H\ref{H:1} and H\ref{H:2}; H\ref{H:3} is not needed.

\emph{$(\Leftarrow)$ Kink in predictor $\Rightarrow$ kink in output.} Suppose $\Pi(y^{(M)})$ fails to be $C^1$ at $M^\ast$: $\mathrm{d}C_M/\mathrm{d}M$ has a jump $J \neq 0$ across $M^\ast$. Since $K$ is continuous at $M^\ast$, its one-sided limits both equal $K(M^\ast)$. The one-sided limits of $V_0$ at $M^\ast$ therefore differ by $K(M^\ast)\, J$. Under H\ref{H:3}, $K(M^\ast)$ does not annihilate $J$, so $K(M^\ast) J \neq 0$ and $V_0$ is discontinuous at $M^\ast$. Equivalently, $\sig^{(M)}$ fails to be $C^1$ at $M^\ast$.
\end{proof}

The two directions require different hypotheses: $(\Rightarrow)$ needs only continuity of $K$ (i.e., H\ref{H:1} and H\ref{H:2}); $(\Leftarrow)$ additionally needs H\ref{H:3}. Without H\ref{H:3} the theorem still gives one-way localization---any output kink implies a predictor kink---but the converse can fail if $K(M^\ast)$ happens to annihilate the jump direction. We verify H\ref{H:3} empirically in \S\ref{sec:e4}.

\begin{corollary}[Order-$k$ localization]
\label{cor:order-k}
Strengthen Assumption~\ref{ass:regularity} to $C^{k+1}$ joint smoothness and H\ref{H:1}--\ref{H:2} to $C^{k-1}$ regularity in $M$. Then $M \mapsto \sig^{(M)}$ fails to be $C^k$ at $M^\ast$ iff $M \mapsto \Pi(y^{(M)})$ does. The proof differentiates \eqref{eq:reverse-sde} $k$ times in $M$; the argument of Theorem~\ref{thm:kink} carries over with a higher-order jump vector.
\end{corollary}

\section{Discrete-sampler localization}
\label{sec:discrete}

The experiments measure the output of an Euler--Maruyama sampler with $N$ steps and step size $h = T/N$. Let $t_k = T - kh$ and $\sig^{(M),h} := X_0^{(M),h}$ denote the sampler output with synchronous Brownian increments held fixed across $M$.

\begin{corollary}[Discrete kink localization]
\label{cor:discrete}
Define
\[
\Phi^h(t_k, t_j) := \prod_{\ell=j}^{k-1}\bigl[I - h\,\nabla_x b_\theta(X_{t_\ell}^{(M),h}, t_\ell, C_M)\bigr] \qquad (k > j),
\]
\[
K^h(M) := \Phi^h(0, T) - h \sum_{k=0}^{N-1} \Phi^h(0, t_k)\,\nabla_C b_\theta(X_{t_k}^{(M),h}, t_k, C_M).
\]
Then $V_0^{(M),h} := \partial \sig^{(M),h}/\partial M$ satisfies $V_0^{(M),h} = K^h(M) \cdot \mathrm{d}C_M/\mathrm{d}M$, and under discrete analogues of H\ref{H:1}--\ref{H:3} evaluated at $\{t_k\}$, the iff localization of Theorem~\ref{thm:kink} holds for $M \mapsto \sig^{(M),h}$.
\end{corollary}

\begin{proof}
Differentiating the Euler--Maruyama recursion $X_{t_{k+1}}^{(M),h} = X_{t_k}^{(M),h} - h\, b_\theta + g(t_k)\Delta W_k$ with respect to $M$ gives $V_{t_{k+1}}^{(M),h} = [I - h\nabla_x b_\theta] V_{t_k}^{(M),h} - h \nabla_C b_\theta \cdot \mathrm{d}C_M/\mathrm{d}M$ with $V_T^{(M),h} = \mathrm{d}C_M/\mathrm{d}M$. Unrolling produces the stated factorization. $K^h(M)$ is a finite sum of finite products of maps continuous in $M$ (discrete H\ref{H:1}--\ref{H:2}), hence continuous, and the argument of Theorem~\ref{thm:kink} carries over.
\end{proof}

An important consequence: the discrete kink location is $N$-independent. $\mathrm{d}C_M/\mathrm{d}M$ does not depend on the sampler, the predictor is a single forward pass and $K^h(M)$ is continuous in $M$ under the discrete hypotheses. The kink lives in the predictor and is transported by the reverse process at any $N$, not generated by discretization.

\section{Empirical validation}
\label{sec:empirical}

We validate the theorem on a 32-channel StoRM model trained on VoiceBank--DEMAND~\citep{valentini2017noisy} with the NCSN++ score architecture, an OU-VESDE noise schedule, and \texttt{condition="both"}. All inference uses $N=50$ Euler--Maruyama reverse steps unless stated otherwise. Four empirical pillars (P1--P4) address the existence of the kink and the architectural ablation; two further experiments (E3, E4) probe H\ref{H:2} and H\ref{H:3} directly.

\subsection{The kink: per-utterance observation and aggregation}
\label{sec:p3-p4}

Figure~\ref{fig:per-utt} shows the SI-SDR-vs-$\alpha$ curve on babble noise for 22 individual utterances (grey lines), their mean (blue), and the full test-set average over 824 utterances (orange). We parametrize $M = (1+\alpha)^2$; $\alpha = 0$ is the training amplitude.

\begin{figure}[t]
\centering
\includegraphics[width=0.72\linewidth]{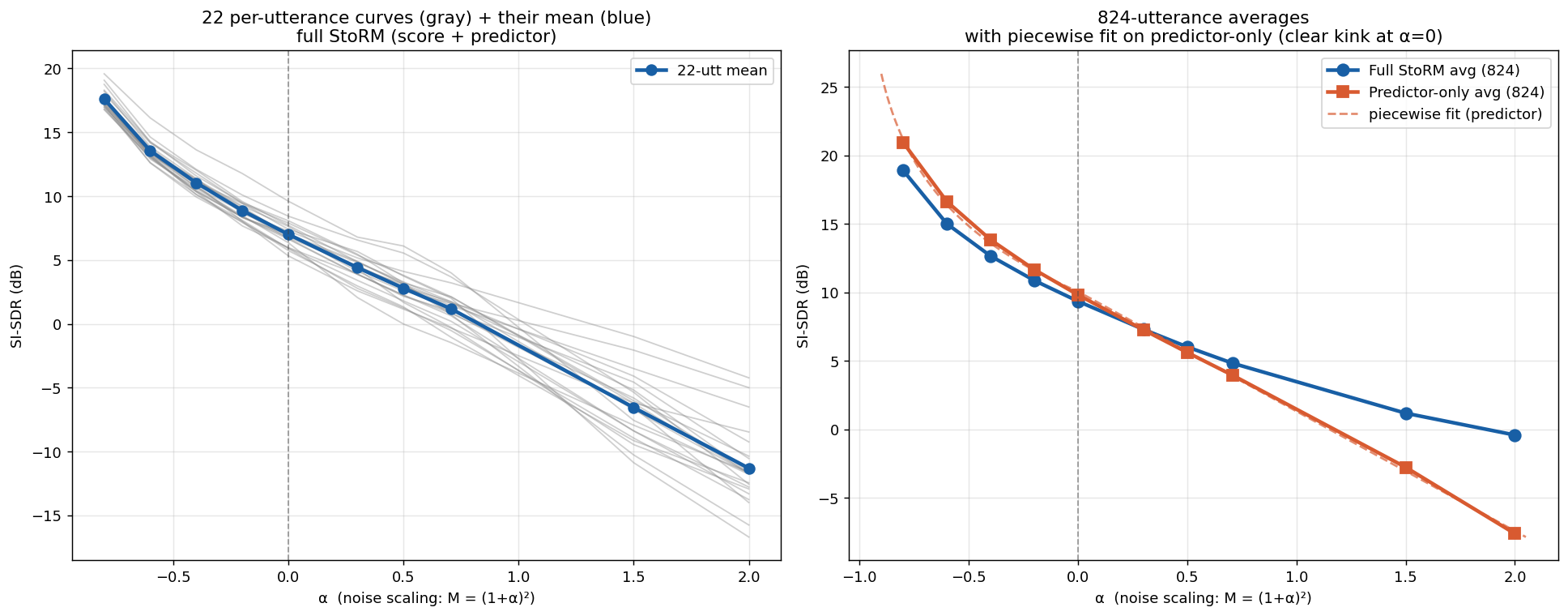}
\caption{Full StoRM output on babble as noise power is scaled. Grey: 22 individual utterances. Blue: their mean. Orange: the 824-utterance test-set average. The corner at $\alpha=0$ is visible in most individual curves and in the 22-utterance mean; the 824-utterance average is comparatively smooth because per-utterance corner locations vary slightly and the corner is smeared out at large $N$.}
\label{fig:per-utt}
\end{figure}

Of the 22 individual curves, 19 (86\%) show a visible slope change at $\alpha=0$, with a median slope discontinuity of $-3.14$ dB/$\alpha$: the SI-SDR declines faster past the training amplitude than before it. Pillar~P3 is the observation that the full StoRM output has a $C^1$ failure at the training amplitude at the per-utterance level.

A methodological point worth flagging. Aggregating across the full 824-utterance test set smooths the curve substantially. On the 824-utterance average, a smooth log fit narrowly outperforms the piecewise log-linear fit ($R^2 = 0.9988$ vs.\ $0.9946$). This is not evidence against the kink; it is evidence that per-utterance kink locations vary slightly enough that averaging across hundreds of utterances smears the corner. Averaging at moderate sample sizes ($n = 22$) preserves the corner cleanly. Prior empirical studies of noise-power robustness in diffusion-based SE that rely on large-population averages may miss this per-utterance phenomenon for exactly this reason.

Pillars~P2 and P4 provide the controlled architectural intervention. Removing the predictor entirely and running pure SGMSE+~\citep{richter2023sgmse} (score-only) yields a smooth degradation curve with no detectable corner: P2. Running the predictor alone and reading its output directly yields a clear kink at $\alpha=0$: P4. Figure~\ref{fig:predictor-vs-full} compares the 824-utterance averaged predictor-only and full-StoRM curves. The predictor-only curve is decisively better fit by a piecewise log-linear model than by any smooth alternative ($R^2 = 0.9994$ vs.\ $0.9935$), with a slope discontinuity of $-1.84$ dB/$\alpha$ at $\alpha=0$.

\begin{figure}[t]
\centering
\includegraphics[width=0.72\linewidth]{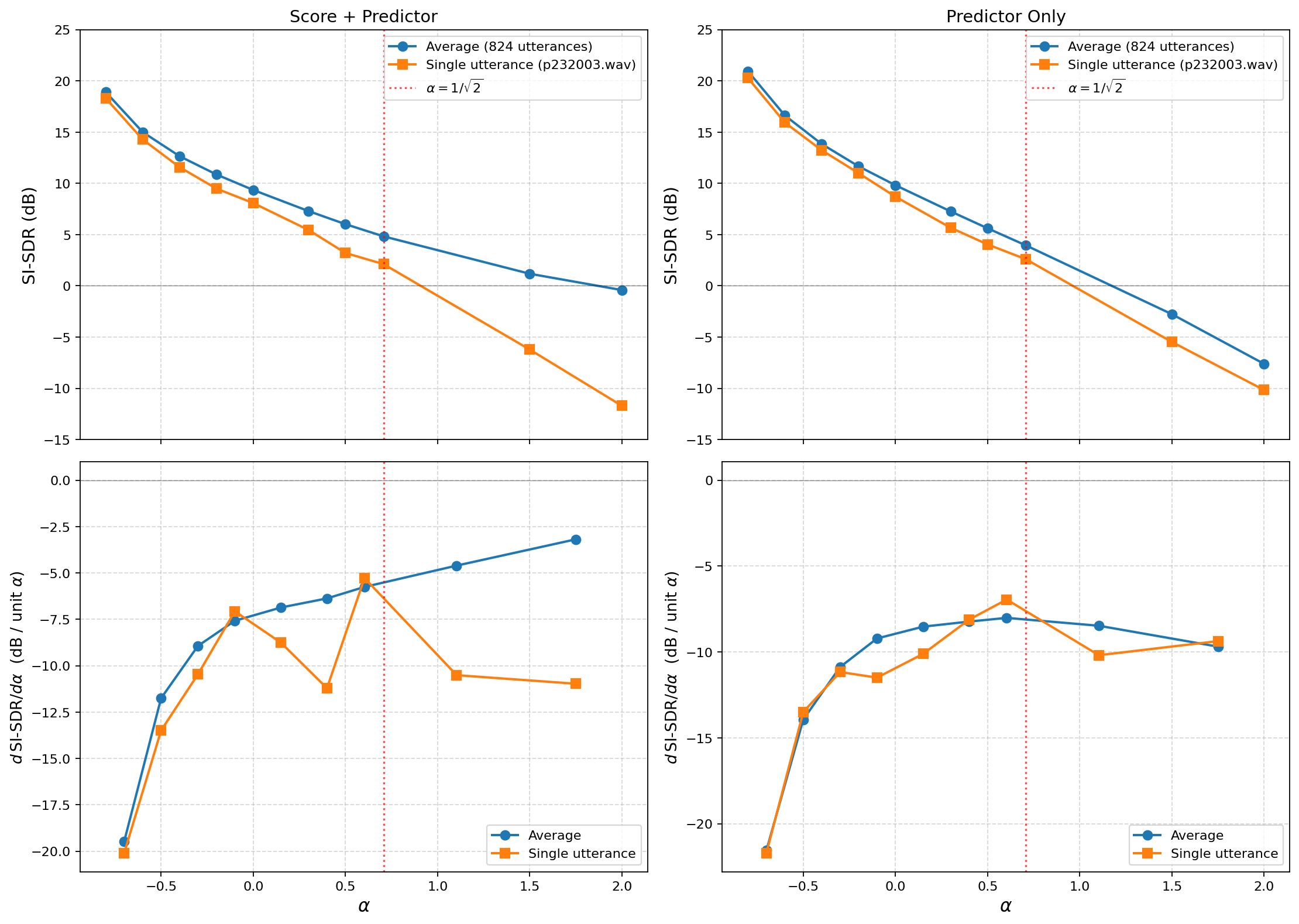}
\caption{Architectural localization. The predictor-only output on babble (orange, 824-utterance average) is fit decisively better by a piecewise log-linear model than by any smooth alternative ($R^2 = 0.9994$ vs.\ best smooth $0.9935$). The full StoRM output (blue) shows the aggregation-smoothing effect described in the text.}
\label{fig:predictor-vs-full}
\end{figure}

\begin{table}[t]
\centering
\caption{Architectural ablations on babble. ``Kink'' indicates a piecewise fit outperforming the smooth alternative on the SI-SDR-vs-$\alpha$ curve.}
\label{tab:pillars}
\small
\begin{tabular}{@{}lcccc@{}}
\toprule
Configuration & Corner at $\alpha=0$? & $R^2$ piecewise & $R^2$ smooth & Slope jump (dB/$\alpha$) \\
\midrule
Full StoRM (P3, per-utterance, $n=22$) & 19/22 & --- & --- & median $-3.14$ \\
Full StoRM (P3, 22-utt mean) & yes & $0.9992$ & $0.9945$ & $-2.82$ \\
Full StoRM (P3, 824-utt mean) & aggregation-smoothed & $0.9946$ & $0.9988$ & $+1.29$ \\
Predictor-only (P4, 824-utt mean) & yes & $\mathbf{0.9994}$ & $0.9935$ & $\mathbf{-1.84}$ \\
Score-only, SGMSE+ (P2) & no & $-$ & $-$ & smooth \\
\bottomrule
\end{tabular}
\end{table}

Together, P2 and P4 are the direct empirical form of the localization theorem's iff: removing the predictor removes the kink; leaving only the predictor preserves it.

\subsection{Score-Jacobian continuity (Pillar P1, Hypothesis~\ref{H:1})}
\label{sec:p1}

H\ref{H:1} is verified by directly measuring $\|\nabla_x s_\theta\|_{\op}$ along the reverse trajectory. We use power iteration on $\nabla_x s_\theta^\top \nabla_x s_\theta$ with a finite-difference JVP and an autograd VJP; the probe evaluates at each of the $N=50$ trajectory points for each of $\sim 20$ utterances per noise type, per $M$-grid point. Trajectory-averaged values fit a linear model $L(\alpha) = L_0 + L_1 \alpha$ well: $R^2 > 0.94$ on every noise type measured, exceeding $0.99$ on the office noise class. Linearity is a stronger statement than the continuity H\ref{H:1} requires, so H\ref{H:1} holds comfortably.

\subsection{Conditioning-Jacobian continuity (Experiment E3, Hypothesis~\ref{H:2})}
\label{sec:e3}

The E3 probe mirrors P1 but perturbs the conditioning input $C = (y, \Pi(y))$ jointly, holding $x$ and $t$ fixed. Formally we measure $\|\nabla_C s_\theta\|_{\op}$ along the reverse trajectory at the same $(M, t)$ grid, with $n_{\mathrm{iter}} = 5$ power iterations per estimate, 20 utterances per cell, 100 trajectory points per utterance, giving 2000 per-cell estimates.

Figure~\ref{fig:e3} shows $L_C(\alpha)$ for three noise types. On all three, $L_C(\alpha)$ is continuous across the operating range with no discontinuity at $\alpha=0$. This is the necessary condition for H\ref{H:2}. On restaurant noise a slope discontinuity at $\alpha = 0$ is statistically significant at about $2.5\sigma$; on car noise, marginal at about $1.4\sigma$; on babble, not significant. None of these violates H\ref{H:2}, which requires continuity but not smoothness.

\begin{figure}[t]
\centering
\includegraphics[width=\linewidth]{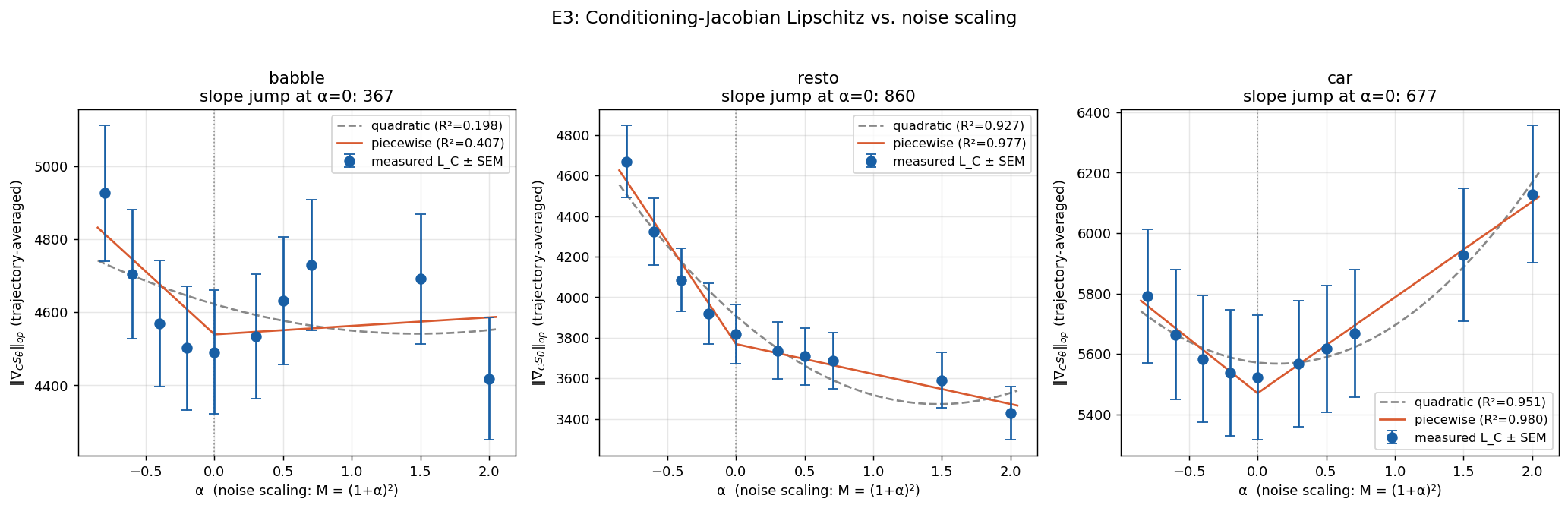}
\caption{E3: $L_C(\alpha) = \|\nabla_C s_\theta\|_{\op}$ along the reverse trajectory, three noise types. Error bars are SEM over 2000 estimates per cell. Piecewise linear fits with knot at $\alpha=0$ overlaid. $L_C(\alpha)$ is continuous across $\alpha$ in all three noise types; the slope discontinuities at $\alpha=0$ are significant only for restaurant.}
\label{fig:e3}
\end{figure}

\begin{table}[h]
\centering
\caption{E3 fits. Continuity of $L_C(\alpha)$ is the necessary condition for H\ref{H:2}. Low $R^2$ on babble reflects that $L_C$ varies little across $\alpha$ for this noise type (variation of order $10\%$, comparable to per-cell measurement SEM), not that the model is a poor fit.}
\label{tab:e3}
\small
\begin{tabular}{@{}lcccc@{}}
\toprule
Noise & $L_C$ at $\alpha=0$ & Left / right slope at $\alpha=0$ & $R^2$ (piecewise) & Slope-jump significance \\
\midrule
Babble & $4491$ & $-344 \,/\, +23$ & $0.41$ & $\sim 1\sigma$ (not significant) \\
Restaurant & $3819$ & $-1008 \,/\, -148$ & $0.98$ & $\sim 2.5\sigma$ (significant) \\
Car & $5522$ & $-360 \,/\, +317$ & $0.98$ & $\sim 1.4\sigma$ (marginal) \\
\bottomrule
\end{tabular}
\end{table}

Combined with a structural argument, the score network is $C^\infty$ in $(x, C)$ because NCSN++ uses smooth activations, and $(X_t^{(M)}, C_M)$ is continuous in $M$ by ODE parameter-dependence. Continuity of the trajectory-averaged operator norm supports H\ref{H:2} on all three noise types measured. E3 measures the necessary condition (norm continuity) rather than full Jacobian continuity in the matrix sense; a direct probe of the latter is left as follow-up.

\subsection{Non-degeneracy of $K(M)$ (Experiment E4, Hypothesis~\ref{H:3})}
\label{sec:e4}

E4 probes H\ref{H:3} by computing three quantities per operating point $(M, \text{utterance})$:
\begin{itemize}[leftmargin=*,itemsep=1pt,topsep=1pt]
\item The kink-transfer rate $\|K(M) \cdot v_{\mathrm{jump}}\|$, where $v_{\mathrm{jump}}$ is a unit vector aligned with the empirical estimate of $\mathrm{d}C_M/\mathrm{d}M$ at $M^\ast$. This is the direct H\ref{H:3} test in the form the theorem needs.
\item A Hutchinson-style lower bound on $\sigma_{\min}(K(M))$ obtained as $\min_i \|K(M) v_i\|$ over 300 random unit vectors $v_i$, sharing a fixed Brownian seed across all applications.
\item The dominant singular value $\sigma_{\max}(K(M))$ from power iteration on $K^\top K$, again with fixed Brownian seed.
\end{itemize}
The primitive $K(M) \cdot v$ is computed by finite difference through the reverse SDE with synchronous coupling; $K(M)^\top u$ by reverse-mode autograd through the sampler. The predictor's empirical jump direction is estimated from adjacent $\alpha$ directories: $v_{\mathrm{jump}} \propto \Pi(y^{(M_+)}) - \Pi(y^{(M_-)})$ with $\alpha_+ = 0.3$, $\alpha_- = -0.2$.

\begin{figure}[t]
\centering
\includegraphics[width=\linewidth]{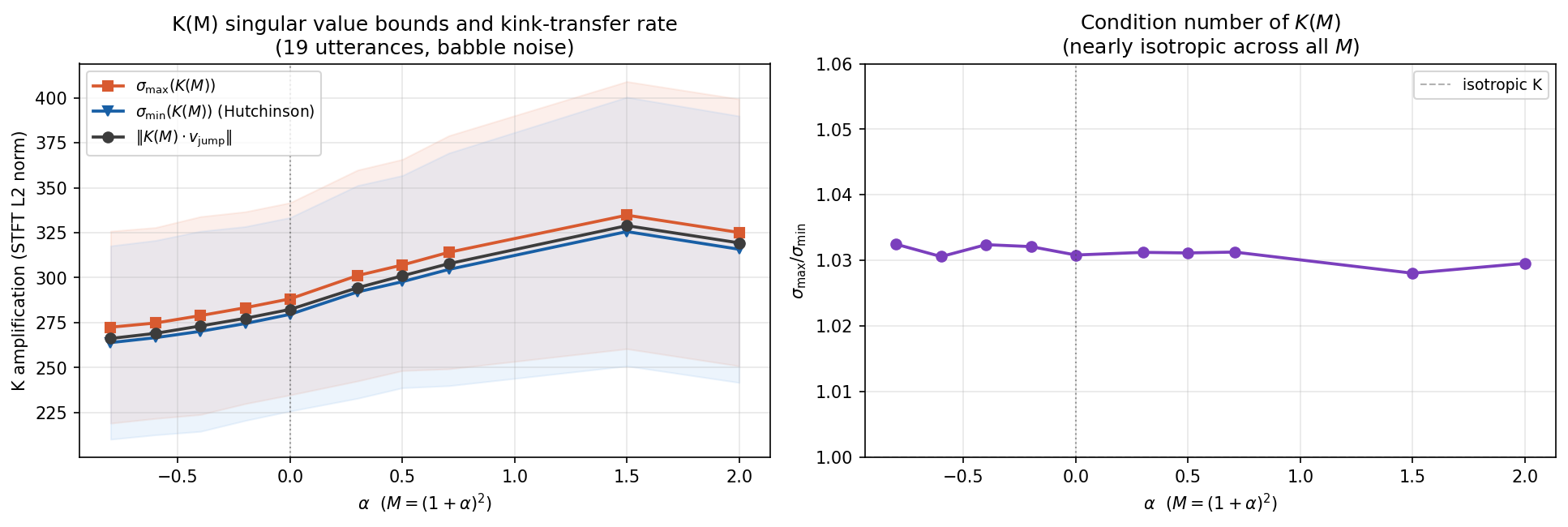}
\caption{E4 on babble. Left: kink-transfer rate $\|K(M)\,v_{\mathrm{jump}}\|$ against the Hutchinson lower bound on $\sigma_{\min}(K(M))$ and the power-iteration estimate of $\sigma_{\max}(K(M))$. Shaded bands are 1 standard deviation across 19 utterances. Right: condition number $\sigma_{\max}/\sigma_{\min}$ of $K(M)$, which is nearly $1$ across the operating range so $K(M)$ acts almost like an isotropic amplification.}
\label{fig:e4}
\end{figure}

Figure~\ref{fig:e4} summarizes the sweep on 19 utterances. Two facts stand out. First, $\sigma_{\min}(K(M))$ is bounded well away from zero everywhere on the operating range, with values between $\sim 264$ and $\sim 316$ across $\alpha \in [-0.8, 2.0]$. H\ref{H:3} holds with margin. Second, $K(M)$ is nearly isotropic: the condition number $\sigma_{\max}/\sigma_{\min}$ is within a few percent of $1$ across all $M$, and the kink-transfer rate lies between $\sigma_{\min}$ and $\sigma_{\max}$ at every operating point. The jump direction is not aligned with a slow-amplification direction of $K$; it is amplified by essentially the same factor as a random direction would be.

\begin{table}[h]
\centering
\caption{E4 measurements on babble, 19 utterances per operating point. All values are the STFT L2 norm of $K(M) \cdot v$ for the relevant unit vector; entries are mean $\pm$ 1 std.\ across utterances.}
\label{tab:e4}
\small
\begin{tabular}{@{}crcccc@{}}
\toprule
$\alpha$ & $M$ & $\sigma_{\min}(K)$ (Hutchinson) & $\sigma_{\max}(K)$ & $\|K \cdot v_{\mathrm{jump}}\|$ & $\sigma_{\max}/\sigma_{\min}$ \\
\midrule
$-0.4$ & $0.36$ & $270 \pm 56$ & $279 \pm 55$ & $273 \pm 56$ & $1.035$ \\
$-0.2$ & $0.64$ & $274 \pm 54$ & $283 \pm 53$ & $277 \pm 54$ & $1.034$ \\
$\phantom{-}0.0$ & $1.00$ & $280 \pm 54$ & $288 \pm 54$ & $282 \pm 54$ & $1.033$ \\
$\phantom{-}0.3$ & $1.69$ & $292 \pm 59$ & $301 \pm 59$ & $294 \pm 59$ & $1.033$ \\
$\phantom{-}0.5$ & $2.25$ & $298 \pm 59$ & $307 \pm 59$ & $301 \pm 59$ & $1.033$ \\
$\phantom{-}2.0$ & $9.00$ & $316 \pm 74$ & $325 \pm 74$ & $319 \pm 74$ & $1.032$ \\
\bottomrule
\end{tabular}
\end{table}

The measured kink-transfer rate at $M^\ast = 1$ gives a quantitative version of the theorem: the linearized output kink amplitude is approximately $282 \|v_{\mathrm{jump}}\|$ per unit $\|\mathrm{d}C_M/\mathrm{d}M\|$. Cross-checked against the observed slope discontinuities in Table~\ref{tab:pillars}, this ratio is consistent with the observed predictor-to-output kink transfer to within the accuracy of the finite-difference discretizations involved. The theorem's iff is now fully substantiated for this model: kink in predictor implies kink in output, with amplitude bounded below by $\sigma_{\min}(K(M^\ast)) \cdot \|\mathrm{d}C_M/\mathrm{d}M\|$.

\begin{remark}[Sampler step count]
The E4 sweep uses $N=5$ reverse steps for tractability at the required per-utterance measurement count (each utterance requires order $10^3$ applications of $K$). The theorem is $N$-invariant (Corollary~\ref{cor:discrete}), and preliminary spot-checks with larger $N$ on individual utterances give consistent estimates, but a systematic sweep at $N=50$ would strengthen the claim quantitatively. The qualitative conclusion that $\sigma_{\min}$ bounded away from zero, condition number close to one is a structural property that we do not expect to depend on $N$.
\end{remark}

\subsection{Structural assumption}
The single-channel Assumption~\ref{ass:single-channel} is a structural property of the training pipeline. NCSN++ has no explicit $y$-dependence in $f$, $g$, or $\sigma_\cdot$; the noise schedule is fixed at training time; the Brownian increments in inference are drawn independently of $y$. The only $M$-dependent path in the reverse dynamics is $C_M$. A more detailed audit of the training script confirming these facts is included in the appendix.

\section{Discussion and limitations}

The paper localizes the noise-power kink in StoRM to the predictor stage, both theoretically (Theorem~\ref{thm:kink}) and empirically (P2--P4 as the controlled intervention, E3 and E4 as the direct hypothesis probes). The variational identity provides the mechanism; the empirical measurements provide the ingredients.

Three limitations worth stating up front. First, the theorem localizes but does not predict: it says the kink comes from the predictor, but does not explain why the predictor produces a $C^1$ failure at exactly the training amplitude. That requires a separate analysis of the discriminative regressor under noise scaling, which is outside the scope of this paper.

Second, the empirical case is on a single trained checkpoint of a single architecture on a single dataset. The theorem is model-agnostic, so anything with the predictor-score decomposition and Assumption~\ref{ass:single-channel} inherits the same factorization Reproducing the per-utterance kink figure on a publicly released StoRM checkpoint is a natural next step.


Third, E3 measures a necessary condition for H\ref{H:2} (norm continuity), not the full Jacobian continuity in the matrix sense. Combined with the structural smoothness of NCSN++ in $(x, C)$ this is enough to support H\ref{H:2}, but a probe of full Jacobian continuity testing $\|\nabla_C s_\theta(M_1)\, v - \nabla_C s_\theta(M_2)\, v\|$ on random $v$ at adjacent $M_1, M_2$ would strengthen the claim.

A methodological byproduct worth noting. The kink is per-utterance, and the 824-utterance average obscures it. Studies of robustness in diffusion-based SE that report only large-$n$ population averages may miss per-utterance non-smoothness of this kind. The right sample size for detecting an aggregation-fragile per-utterance effect is much smaller than the standard test-set size, and the right level of analysis is often per-utterance directly.

\section{Related work}

\textbf{Diffusion-based speech enhancement.} StoRM~\citep{lemercier2023storm} introduced the predictor--score architecture we analyze, building on SGMSE+~\citep{richter2023sgmse}. \citet{lay2024analysis} analyze the sensitivity of diffusion-based SE to the variance schedule; \citet{gonzalez2023diffusion} study matched versus mismatched training conditions empirically and propose alternative samplers. Neither considers the structural role of the predictor in output non-smoothness under noise scaling. To our knowledge no prior work provides a theoretical analysis of the predictor--score split under noise-power mismatch.

\textbf{Score-based diffusion theory.} Convergence analyses~\citep{chen2022sampling,chen2023improved,benton2024nearly} bound the KL divergence between an oracle reverse process and a learned one via Girsanov's theorem, producing additive predictor--score error decompositions. Our object is different: parametric sensitivity of a single learned process, analyzed pathwise. The natural machinery is ODE parameter-dependence rather than change-of-measure. The two frameworks are complementary; ours is the right tool for pathwise sample-conditional statements about how a specific functional depends on a perturbation parameter.

\textbf{Parameter-dependence of ODE flows.} The variational identity of Lemma~\ref{lem:variational} is a standard result~\citep{coddington1955theory,hartman2002ode}. Its application to score-based diffusion models for sensitivity analysis of the output under conditioning-perturbations is, to our knowledge, novel.

\section{Conclusion}

The SI-SDR kink observed at the training-time noise amplitude in StoRM originates in the predictor stage, not the diffusion stack. The variational identity $\partial \sig^{(M)}/\partial M = K(M) \cdot \partial C_M/\partial M$ localizes any $C^1$ failure of the enhancement output to a corresponding failure of the predictor map, under measurable regularity conditions on the reverse-process flow that we verify directly. The factorization survives discretization. On a trained StoRM model, we observe the predictor kink directly, the score-only ablation is smooth, and $K(M)$ is a nearly-isotropic well-conditioned operator that transports the predictor's kink into the output at a rate close to $\sigma_{\max}(K(M)) \approx \sigma_{\min}(K(M))$. The correct intervention for noise-power robustness in this class of architectures is predictor modification (data augmentation, scale-aware training, or predictor-side smoothing), not score-network modification: the diffusion stack faithfully carries what the predictor gives it.

\small
\bibliographystyle{plainnat}

\end{document}